\documentclass[pdflatex,sn-nature]{sn-jnl}

\usepackage{graphicx}%
\usepackage{multirow}%
\usepackage{amsmath,amssymb,amsfonts}%
\usepackage{amsthm}%
\usepackage{mathrsfs}%
\usepackage[title]{appendix}%
\usepackage{xcolor}%
\usepackage{textcomp}%
\usepackage{manyfoot}%
\usepackage{booktabs}%
\usepackage{algorithm}%
\usepackage{algorithmicx}%
\usepackage{algpseudocode}%
\usepackage{listings}%
\usepackage{array}

\theoremstyle{thmstyleone}%
\newtheorem{theorem}{Theorem}

\theoremstyle{thmstyletwo}%
\newtheorem{remark}{Remark}

\begin{document}

\title[Minimum Energy Cruise of All-Electric Aircraft with Applications to Advanced Air Mobility]{Minimum Energy Cruise of All-Electric Aircraft with Applications to Advanced Air Mobility}

\author*[1]{\fnm{Steven} \sur{Li}}\email{steven.li@concordia.ca}

\author[1]{\fnm{Luis} \sur{Rodrigues}}\email{luis.rodrigues@concordia.ca}

\affil*[1]{\orgdiv{Department of Electrical and Computer Engineering}, \orgname{Concordia University}, \city{Montréal}, \state{Québec}, \country{Canada}}

\abstract{Electrified propulsion is expected to play an important role in the sustainable development of Advanced Air Mobility (AAM). However, the limited energy density of batteries motivates the need to minimize energy consumption during flight. This paper studies the minimum total energy problem for an all-electric aircraft in steady cruise flight. The problem is formulated as an optimal control problem in which the cruise airspeed and final cruise time are optimization variables. The battery supply voltage is modeled as an affine function of the battery charge. Pontryagin's Minimum Principle is used to derive the necessary and sufficient conditions for optimality, from which closed-form expressions for the optimal cruise airspeed and optimal final cruise time are obtained. Additional analytical conditions are derived that determine when all-electric operation is feasible, one of which is that sufficient electric charge must be available. Numerical simulations based on the BETA Technologies CX300 all-electric aircraft and a representative AAM scenario illustrate how the aircraft weight, cruising altitude, electrical system efficiency, and initial battery charge influence the optimal airspeed and the feasibility of all-electric cruise.}

\maketitle

\section{Introduction}
For the sustainable development of Advanced Air Mobility (AAM), replacing conventional fuel-based propulsion with electrical energy offers the potential to reduce emissions in aviation systems while enabling quieter operations. In many AAM applications, aircraft are expected to operate fully electric propulsion systems powered by onboard batteries \cite{goyal2022advanced,al2023advanced,bridgelall2024aircraft,raza2025advanced}. However, the amount of usable energy carried onboard the aircraft is limited by the battery capacity. Consequently, the efficient management of onboard electrical energy is a critical consideration in the operation of all-electric aircraft. Cruise flight accounts for a significant portion of total mission energy consumption, meaning that the choice of cruise operating conditions can substantially influence the feasibility of all-electric aircraft missions. This consideration is especially relevant for AAM operations, where aircraft are expected to perform frequent short-range missions in low-altitude airspace, thereby making energy-efficient cruise strategies important for mission planning and operational feasibility.

Several studies have investigated the performance characteristics of all-electric aircraft and the influence of operating conditions on mission energy consumption. 
Analytical expressions for estimating the range and endurance of battery-powered aircraft were derived in \cite{traub2011range}, followed by the derivation of the optimal battery weight fraction that maximizes range and endurance in \cite{traub2016optimal}. Methods for estimating energy consumption during climb and cruise phases were presented in \cite{ma2017method}, while the work in \cite{barufaldi2023optimal} investigated maximum range and endurance for all-electric aircraft with fixed-pitch propellers. The influence of cruise conditions and operational uncertainties on battery state-of-charge was examined using an optimal control framework in \cite{pradeep2021parametric}, demonstrating that predicted battery usage depends significantly on the chosen battery model. The selection of range-optimal speeds under wind conditions while accounting for propulsion and battery dynamics was investigated in \cite{spark2024optimizing}. In addition, the feasibility of all-electric aircraft operations under extreme environmental conditions was analyzed in \cite{adu2026system} using an integrated aircraft energy model.

Another line of research has focused on the design of all-electric aircraft and its subsystems, including batteries, propulsion systems, and airframe configurations. For instance, conceptual design requirements for small electric aircraft were investigated in \cite{patterson2014conceptual}, while the weight of the electrical propulsion system was  minimized through the analysis of different power distribution architectures in \cite{ebersberger2023power}. Reviews of emerging technologies such as cryogenic power systems for electric aircraft have also been presented in \cite{elwakeel2025review}. Conceptual feasibility studies of large passenger all-electric aircraft have been conducted in \cite{zilliac2024feasibility}, and the performance of different electric motor configurations for electric propulsion systems has been analyzed in \cite{chen2025comparison}. 

Despite the growing body of work on electric aircraft design and performance, relatively few studies have addressed the energy minimization problem for all-electric aircraft. Most of these studies focused on the related problem of minimizing direct operating costs (DOC). Analytical expressions for the cruise airspeed that minimize the DOC of an all-electric aircraft using a constant-voltage battery model were derived in \cite{kaptsov2017flight}. The minimum-energy problem for a multirotor electric aircraft was formulated as an optimal control problem and studied numerically in \cite{pradeep2022wind}. Similarly, the DOC minimization problem for all-electric aircraft was formulated as an optimal control problem using several battery models, including constant-voltage, empirical circuit, and electrochemical models, and solved numerically in \cite{wang2022energy}. Analytical solutions for the DOC minimization problem were later derived for both hybrid-electric and all-electric aircraft and applied to path planning problems in AAM contexts \cite{li2023optimal}. More recently, the DOC minimization problem was studied using a variable cost index formulation, from which expressions for the optimal cruise airspeed and final cruise time were derived by the authors of \cite{e2025unified}. Table \ref{tab:literature_comparison} presents a comparison of existing studies on energy optimization for all-electric aircraft. Only studies addressing at least one of the criteria listed in the table are included. As shown in Table \ref{tab:literature_comparison}, prior work has either relied on numerical optimization approaches, did not address the feasibility of all-electric operation based on battery charge availability, or assumed simplified battery models that do not account for the dependence of the supply voltage on the battery charge. The present work addresses these limitations. The main contributions of this paper are as follows:
\begin{itemize}
    \item Closed-form expressions are derived for the optimal airspeed and the optimal final cruise time that minimize the total energy of an all-electric aircraft in steady cruise flight, while modeling the battery supply voltage as an affine function of the battery charge.
    \item Analytical feasibility conditions for all-electric operation are obtained that explicitly determine whether sufficient battery charge is available to complete the cruise segment.
\end{itemize}
\begin{table}[h]
\centering
\caption{Comparison of All-Electric Aircraft Energy Optimization Studies}
\label{tab:literature_comparison}
\begin{tabular}{c|>{\centering\arraybackslash}p{2cm}|>{\centering\arraybackslash}p{2cm}|>{\centering\arraybackslash}p{2cm}|>{\centering\arraybackslash}p{2cm}}
\hline
\textbf{Ref.} & \textbf{Analytical Solutions} & \textbf{Variable Airspeed} & \textbf{Charge Feasibility} & \textbf{Nonconstant Voltage} \\
\hline
\cite{kaptsov2017flight} & \textbf{Yes} & \textbf{Yes} & No & No \\
\cite{pradeep2022wind} & No & \textbf{Yes} & No & No \\
\cite{wang2022energy} & No & \textbf{Yes} & No & \textbf{Yes} \\
\cite{li2023optimal} & \textbf{Yes} & \textbf{Yes} & No & No \\
\cite{e2025unified} & \textbf{Yes} & \textbf{Yes} & No & No \\
Proposed & \textbf{Yes} & \textbf{Yes} & \textbf{Yes} & \textbf{Yes} \\
\hline
\end{tabular}
\end{table}
The remainder of this paper is organized as follows. Section II introduces the system model of the all-electric aircraft. Section III formulates the energy minimization problem in cruise and derives the feasibility conditions that determine whether sufficient battery charge is available. Section IV presents numerical simulations that illustrate how key system parameters influence the optimal cruise airspeed and the feasibility of all-electric operation, and discusses the results in the context of AAM.

\section{System Model}
This section describes the system model of an all-electric aircraft. The following assumptions are made for the system model:
\begin{enumerate}
    \item The aircraft is assumed to fly in steady cruise at constant altitude, and the atmospheric air density $\rho$ is constant. \label{ass:steady_cruise}
    \item Wind effects are assumed to be small, so they can be neglected. \label{ass:no_wind}
    \item The flight Mach number is assumed to remain below the drag divergence Mach number, thereby avoiding the sharp increase in the drag coefficient associated with shock waves \cite{Anderson_air_perf}. \label{ass:drag_polar}
    \item The aircraft is assumed to operate within its flight envelope. \label{ass:flight_envelope}
    \item \label{ass:constant_voltage}The supply voltage $U$ of the electric motors is assumed to remain within the nominal region of the battery discharge cycle (see Figure \ref{fig:battery_model}), and is therefore treated as an affine function of the battery charge during cruise. 
    \item The internal resistance of the battery is assumed to be small and independent of the supply current magnitude, with its effects incorporated into the system through a reduction in the electrical system efficiency $\eta$. \label{ass:small_resistance}
    \item The battery capacity is assumed to be independent of the discharge current (no Peukert effect \cite{doerffel2006critical}). \label{ass:no_peukert}
    \item The temperature of the batteries remains within an acceptable range, and thermal effects within the battery system are neglected. \label{ass:no_thermal}
\end{enumerate}
\begin{figure}
        \centering
        \includegraphics[scale=0.37]{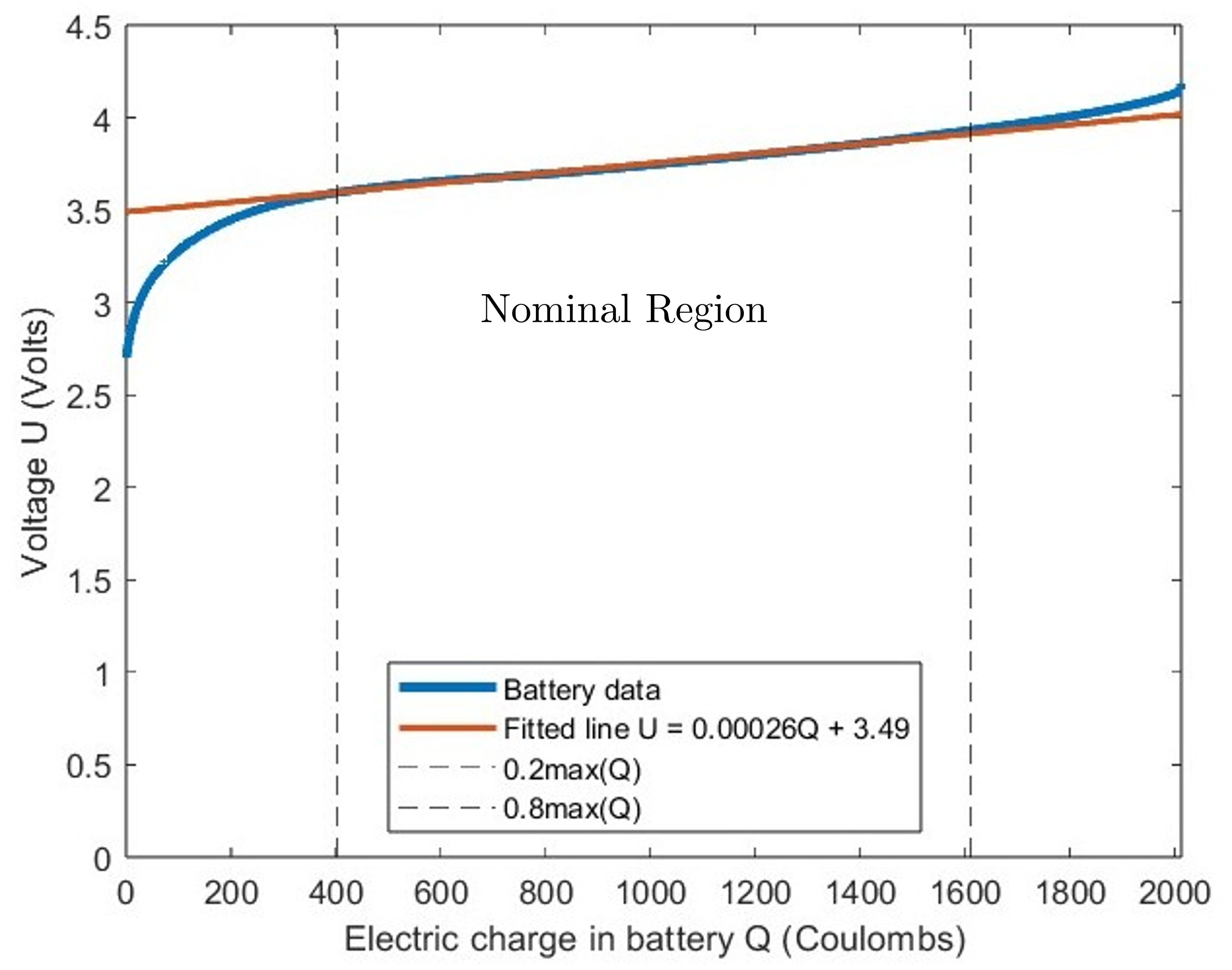}
        \caption{Example of a battery discharge cycle obtained from the Oxford Battery Degradation Dataset 1 \cite{howey2017oxford}.}
        \label{fig:battery_model}
    \end{figure}
\begin{remark}
    Regarding Assumption \ref{ass:constant_voltage}, it is standard practice to operate batteries within a restricted state-of-charge window, typically between $20\%$ and $80\%$ of the maximum capacity to mitigate battery degradation \cite{gao2018aging}.
\end{remark}
\begin{remark}
    Regarding Assumption \ref{ass:no_thermal}, modern aircraft battery packs are equipped with active thermal management systems that regulate battery cell temperature. Therefore, provided the system operates within its design envelope, temperature variations should remain sufficiently small to justify neglecting thermal effects within the battery system.
\end{remark}
\subsection{Flight Dynamics Model}
Let $x(t)$ denote the horizontal position of the aircraft and $v$ its airspeed. Given Assumption \ref{ass:no_wind}, the rate of distance traveled of the aircraft is
\begin{equation}\label{dotx}
    \dot{x}(t) = v
\end{equation}
and $v_{stall} < v < v_{max}$ with $v_{stall} > 0$ being the stall speed of the aircraft given by
\begin{equation}
    v_{stall} = \sqrt{\frac{2W}{\rho SC_{L,max}}}
\end{equation}
where $W$ is the magnitude of the aircraft weight, $\rho$ is the atmospheric air density, $S$ is the wing surface area, and $C_{L,max}$ is the maximum lift coefficient, and with 
\begin{equation}
    v_{max} = \min(v_{div}, v_{max,rtd})   
\end{equation}
where $v_{div}$ is the airspeed associated with the drag divergence Mach number, and $v_{max,rtd}$ is the rated maximum operating airspeed of the aircraft. In addition, following Assumption \ref{ass:steady_cruise},
\begin{equation}\label{TequalD}
    T = D
\end{equation}
\begin{equation}\label{LequalW}
    L = W
\end{equation}
where $T$, $D$, and $L$ denote the magnitudes of the thrust, drag, and lift forces, respectively. Under Assumption \ref{ass:drag_polar}, and using \eqref{LequalW}, the lift and drag coefficients are
\begin{equation}\label{CL}
    C_L = \frac{2L}{\rho Sv^2} = \frac{2W}{\rho Sv^2}
\end{equation}
\begin{equation}\label{CD}
    C_D = C_{D,0} + C_{D,2}C_L^2
\end{equation}
where $C_{D,0}$ and $C_{D,2}$ denote the profile and induced drag coefficients, respectively. The resulting magnitude of the drag force is therefore defined as
\begin{equation}\label{drag}
	D = \frac{1}{2}\rho Sv^2C_D = \frac{1}{2}C_{D,0}\rho Sv^2 + \frac{2C_{D,2}W^2}{\rho Sv^2}
\end{equation}
The following partial derivatives will be used later in the proofs of Theorem 1.
\begin{equation}\label{Dv}
    D_v = \frac{\partial D}{\partial v} = C_{D,0}\rho Sv - \frac{4C_{D,2}W^2}{\rho Sv^3}
\end{equation}
\begin{equation}\label{Dvv}
    D_{vv} = \frac{\partial^2 D}{\partial v^2} = C_{D,0}\rho S + \frac{12C_{D,2}W^2}{\rho Sv^4}
\end{equation}
The minimum-drag airspeed is obtained by solving $D_v = 0$ and is given by
\begin{equation}\label{vDmin}
    v_{D,min} = \sqrt{\frac{2W}{\rho S}\sqrt{\frac{C_{D,2}}{C_{D,0}}}}
\end{equation}
and is constant.

\subsection{Electricity Consumption Model}
Given Assumption \ref{ass:constant_voltage}, the supply voltage is defined as
\begin{equation}\label{UofQ}
    U(t) = aQ(t) + b
\end{equation}
where $a \geq 0$ and $b > 0$ are linear and affine constant coefficients, respectively, and $Q(t)$ denotes the electric charge stored in the batteries, constrained by $Q_{min} < Q(t) < Q_{max}$, with $Q_{min} = 0.2Q_{full}$, $Q_{max} = 0.8Q_{full}$, and $Q_{full}$ is the rated maximum charge of the battery system. The power required for electric propulsion must equal the available battery power (after accounting for losses in the conversion to mechanical energy), expressed as
\begin{equation}
    Tv = \eta U(t)i(t) = -\eta U(t)\dot{Q}(t)
\end{equation}
where $i(t)$ is the current drawn by the electric propulsion system and $\eta > 0$ is the electrical system efficiency. The negative sign indicates that the current is drawn from the battery system. Under Assumptions \ref{ass:steady_cruise} and \ref{ass:small_resistance} to \ref{ass:no_thermal}, and using \eqref{TequalD} along with Faraday's law, the electrical current drawn from the battery system is given by
\begin{equation}\label{dotQsimp}
    \dot{Q}(t) = -\frac{Dv}{\eta U(t)} = -i(t)
\end{equation}
Substituting \eqref{UofQ} into \eqref{dotQsimp} yields
\begin{equation}\label{dotQ}
    \dot{Q}(t) = -\frac{Dv}{\eta (aQ(t) + b)} = -i(t)
\end{equation}

\section{Minimum Total Energy}
The total energy of an all-electric aircraft is defined as
\begin{equation}\label{TEC}
	J(t) = \int_{t_0}^{t_f}{U(t)i(t)dt}
\end{equation}
where the times $t_0$ and $t_f$ correspond to the start and end of the cruise phase, respectively. The minimum energy problem is therefore formulated as the following optimal control problem (OCP):
\begin{equation}\label{OCP}
    \begin{aligned}
	    J^*(t) = &\min_{v(t),t_f}{\int_{t_0}^{t_f}{U(t)i(t)dt}}\\
	    &s.t.\\
	    &\dot{x}(t) = v\\
        &\dot{Q}(t) = -\frac{Dv}{\eta U(t)} = -i(t)\\
	    &D = \frac{1}{2}C_{D,0}\rho Sv^2 + \frac{2C_{D,2}W^2}{\rho Sv^2}\\
        &U(t) = aQ(t) + b > 0\\
	    &Q_{min} < Q(t) < Q_{max}\\
        &x(t_0) = x_0, \ x(t_f) = x_f \\
	    &Q(t_0) = Q_0\\
	    &v_{stall} < v < v_{max}
    \end{aligned}
\end{equation}
where $x_0$ denotes the initial cruise position of the aircraft, $Q_0$ is the initial electric charge at the start of the cruise phase, and $x_f$ is the position of the aircraft at the end of cruise.
\begin{theorem}
    Given $a$, $b$, $x_0$, $x_f$, $Q_0$, $Q_{min}$, and $Q_{max}$, the optimal airspeed is
    \begin{equation}\label{vopt_CI=0}
        v^* = \sqrt{\frac{2W}{\rho S}\sqrt{\frac{C_{D,2}}{C_{D,0}}}} = v_{D,min}
    \end{equation}
    with optimal final cruise time
    \begin{equation}
        t_f = t_0 + \frac{x_f - x_0}{v^*}
    \end{equation}
    if $v_{stall} < v^* < v_{max}$ and if $Q_0 < Q_{max}$, $Q(t_f) > Q_{min}$, where
    \begin{equation}
        Q(t_f) = \begin{cases}
            Q_0 - \frac{2W\sqrt{C_{D,0}C_{D,2}}}{\eta b}(x_f - x_0)&if\ a = 0\\
            \frac{-b + \sqrt{b^2 - 2aZ(t_f)}}{a}&if\ a > 0
        \end{cases}
    \end{equation}
    with
    \begin{equation}
        Z(t_f) = \frac{2W\sqrt{C_{D,0}C_{D,2}}}{\eta}\left(x_f - x_0\right) - \frac{aQ_0^2}{2} - bQ_0 < 0
    \end{equation}
\end{theorem}

\begin{proof}
     The Hamiltonian of the optimal control problem \eqref{OCP} is
    \begin{equation}\label{Hdocfull}
        \begin{aligned}
            H = Ui + J^*_x\dot{x} + J^*_Q\dot{Q}
        \end{aligned}
    \end{equation}
    where $J^*_x(t)$ and $J^*_Q(t)$ denote the costate variables associated with the states $x(t)$ and $Q(t)$. Substituting the state equations \eqref{dotx} and \eqref{dotQ}, along with \eqref{UofQ} into \eqref{Hdocfull} yields
    \begin{equation}\label{H}
    \begin{aligned}
        H = &\ J^*_xv + \frac{Dv}{\eta} - J^*_Q\frac{Dv}{\eta (aQ+b)}
    \end{aligned}
    \end{equation}
    The terminal constraint is
    \begin{equation*}
        \Psi(t_f) = x(t_f) - x_f = 0
    \end{equation*}
    The transversality conditions are
    \begin{equation}\label{Htf}
        H(t_f) = -\nu_t\frac{\partial\Psi}{\partial t}\bigg\vert_{t_f} = 0
    \end{equation}
    \begin{equation}\label{JQf}
        J^*_Q(t_f) = \nu_Q\frac{\partial\Psi}{\partial Q}\bigg\vert_{t_f} = 0
    \end{equation}
    where $\nu_t$ and $\nu_Q$ are Lagrange multipliers. From Hamilton's canonical equations,
    \begin{equation}\label{dotJx}
        \dot{J}^*_x(t) = - \frac{\partial H}{\partial x} = 0
    \end{equation}
    \begin{equation}\label{dotJQ}
        \dot{J}^*_Q(t) = - \frac{\partial H}{\partial Q} = -\frac{J^*_Qa Dv}{\eta(aQ + b)^2} = \frac{J^*_Qa\dot{Q}}{aQ+b}
    \end{equation}
    Equation \eqref{dotJx} implies that $J^*_x(t)$ is constant. The equation \eqref{dotJQ} can be written as a separable differential equation for $J^*_Q(t) \neq 0$:
    \begin{equation*}
        \frac{1}{J^*_Q}\frac{dJ^*_Q}{dt} = \frac{a}{ aQ+b}\frac{dQ}{dt}
    \end{equation*}
    Solving this separable equation results in
    \begin{equation*}
        J^*_Q(t) = K_1(aQ(t) + b)
    \end{equation*}
    where $K_1$ is a constant of integration. Since $(aQ(t) + b) > 0$ and $J^*_Q(t_f) = 0$, then $K_1 = 0$. Therefore, $\dot{J}^*_Q(t) = 0$ and
    \begin{equation}\label{JQ=0}
        J^*_Q(t) = J^*_Q(t_f) = 0\quad\forall t\in[t_0,t_f]
    \end{equation}
    Since the Hamiltonian does not depend explicitly on time, 
    \begin{equation}\label{dotH}
	    \dot{H} = \frac{\partial H}{\partial t} = 0
    \end{equation}
    so $H(t)$ remains constant. Using \eqref{Htf}, 
    \begin{equation}\label{Hforallt}
	    H(t) = H(t_f)=0,~\forall t\in \left[t_0, t_f\right]
    \end{equation}
    From \eqref{H} and the Pontryagin's Minimum Principle \cite{pontryagin1962}, the first-order and second-order necessary conditions for a minimum with respect to $v$ are
    \begin{equation}\label{Hv}
    \begin{aligned}
	    \frac{\partial H}{\partial v} =\ & J^*_x + \left(\frac{1}{\eta} - \frac{J^*_Q}{\eta (aQ + b)}\right)\left(D_vv+ D\right) = 0
    \end{aligned}
    \end{equation}
    \begin{equation}\label{Hvv}
    \begin{aligned}
        \frac{\partial^2H}{\partial v^2} =\ & \left(\frac{1}{\eta} - \frac{J^*_Q}{\eta (aQ + b)}\right)\left(D_{vv}v+2D_v\right)
    \end{aligned}
    \end{equation}
    Since $J^*_x(t)$ is constant, equation \eqref{Hv} yields
    \begin{equation}\label{Jx}
    \begin{aligned}
        J^*_x =\ &-\left(\frac{1}{\eta} - \frac{J^*_Q}{\eta (aQ + b)}\right)\left(D_vv+ D\right)
    \end{aligned}
    \end{equation}
    Using \eqref{JQ=0} together with \eqref{Hforallt}, equation \eqref{H} becomes
    \begin{equation}\label{H=0,JQ=0}
        H = J_x^*v + \frac{1}{\eta}Dv = 0
    \end{equation}
    Substituting \eqref{JQ=0} into \eqref{Hv} and \eqref{Hvv} gives
    \begin{equation}\label{Hv_sigma_neg}
        \frac{\partial H}{\partial v} = \frac{1}{\eta}(D_vv+D) + J_x^* = 0
    \end{equation}    \begin{equation}\label{Hvv_sigma_neg}
        \frac{\partial^2 H}{\partial v^2} = \frac{1}{\eta}(D_{vv}v + 2D_v) \geq 0
    \end{equation}
    From \eqref{Jx} with \eqref{JQ=0}, 
    \begin{equation*}
        J^*_x = -\frac{1}{\eta}(D_vv + D)
    \end{equation*}
    and substituting this expression into \eqref{H=0,JQ=0} yields
    \begin{equation}\label{HCI=0_beta=1}
        \frac{1}{\eta}D_vv^2 = 0
    \end{equation}
    Since $v > 0$ and $\eta > 0 $, equation \eqref{HCI=0_beta=1} implies $D_v = 0$. When $D_v = 0$, and since $D_{vv} > 0$, $v > 0$, and $\eta > 0$,  
    \begin{equation}\label{Huu>0}
        \frac{\partial^2 H}{\partial v^2}= \frac{1}{\eta}D_{vv}v > 0
    \end{equation}
    which implies that the Hamiltonian is strictly convex with respect to $v$. Consequently, the solution of $D_v = 0$ uniquely minimizes the Hamiltonian. Solving $D_v = 0$ yields
    the airspeed\begin{equation}\label{vopt_sigma_neg_CI_0}
        v^* = \sqrt{\frac{2W}{\rho S}\sqrt{\frac{C_{D,2}}{C_{D,0}}}} = v_{D,min}
    \end{equation}
    Since the Hamilton-Jacobi-Bellman equation \eqref{H=0,JQ=0} is satisfied, then if the constraint $v_{stall} < v^* < v_{max}$ is satisfied, the conditions of Theorem 5-12 of \cite{athans2013optimal} hold and $v^*(t)$ is the optimal airspeed.
    Finally, since $\dot{v} = 0$, the optimal airspeed is constant, which implies that the drag $D$ and its partial derivatives \eqref{Dv} and \eqref{Dvv} are also constant.
    The optimal final cruise time is given by
    \begin{equation}\label{tf_opt}
        t_f = t_0 + \frac{x_f - x_0}{v^*}
    \end{equation}
    Since $D > 0$, $v > 0$, $\eta > 0$, and $(aQ(t) + b) > 0$, it follows that $\dot{Q}(t) < 0$. Hence, $Q(t)$ is strictly decreasing on $[t_0,t_f]$. Consequently, the minimum value of $Q(t)$ is attained at $t_f$, and the maximum value at $t_0$. Thus, the constraint $Q(t) < Q_{max}$ is satisfied for all $t\in[t_0,t_f]$ if and only if $Q_0 < Q_{max}$, while the constraint $Q(t) > Q_{min}$ is satisfied for all $t\in [t_0,t_f]$ if and only if $Q(t_f) > Q_{min}$. Since equation \eqref{dotQ} is separable, this implies that
    \begin{equation}\label{dotQ_separable}
        \int{(aQ(t) + b)dQ} = \int{-\frac{Dv^*}{\eta}dt}
    \end{equation}
    Since $D$ and $v^*$ are constant over $[t_0,t_f]$, evaluating the above integrals yields
    \begin{equation}\label{quadratic_Q(t)}
        \frac{aQ(t)^2}{2} + bQ = -\frac{Dv^*}{\eta}t + K_2
    \end{equation}
    where $K_2$ is a constant of integration. Given $Q(t_0) = Q_0$, and since $D$ and $v$ are constant, then
    \begin{equation*}
        K_2 = \frac{aQ_0^2}{2} + bQ_0 + \frac{Dv^*}{\eta}t_0
    \end{equation*}
    For the special case $a = 0$, corresponding to constant supply voltage, the solution for $Q(t)$ is
    \begin{equation}\label{Q(t)_constant_voltage}
        Q(t) = Q_0 -\frac{Dv^*}{\eta b}(t-t_0)
    \end{equation}
    This function is a linearly decreasing with $t$, which means that its minimum value in the interval $[t_0,t_f]$ is at $t_f$. Therefore, for $Q(t) > Q_{min}$ to be satisfied, it is sufficient that $Q(t_f) > Q_{min}$. As a result, substituting \eqref{tf_opt} into \eqref{Q(t)_constant_voltage}, we obtain
    \begin{equation*}
        Q_0 - \frac{D}{\eta b}(x_f - x_0) > Q_{min}
    \end{equation*}
    Replacing $D$ in the above inequality with the result of using \eqref{vopt_sigma_neg_CI_0} into \eqref{drag} yields
    \begin{equation*}
        Q_0 - \frac{2W\sqrt{C_{D,0}C_{D,2}}}{\eta b}(x_f - x_0) > Q_{min}
    \end{equation*}
    For $a> 0$, \eqref{quadratic_Q(t)} is quadratic in $Q(t)$. By Descartes' rule of signs \cite{gomes_implicit}, this equation has exactly one positive real root, provided that 
    \begin{equation}\label{Q(t)>0_cond}
        \frac{Dv^*}{\eta}(t-t_0) - \frac{aQ_0^2}{2} - bQ_0 < 0
    \end{equation}
    Define
    \begin{equation}\label{Z(t)}
        Z(t) = \frac{Dv^*}{\eta}(t-t_0) - \frac{aQ_0^2}{2} - bQ_0
    \end{equation}
    Equation \eqref{Z(t)} is a linearly increasing function of $t$, and thus attains its maximum value over $[t_0,t_f]$ at $t_f$. Therefore, if \eqref{Q(t)>0_cond} holds at $t_f$, then it holds for all $t\in[t_0,t_f]$. Evaluating \eqref{Z(t)} at $t_f$ and using \eqref{tf_opt} yields
    \begin{equation}\label{Q(t_f)>0_cond}
        Z(t_f) = \frac{D}{\eta}\left(x_f - x_0\right) - \frac{aQ_0^2}{2} - bQ_0 < 0
    \end{equation}
    Substituting \eqref{vopt_sigma_neg_CI_0} into  \eqref{drag}, and using \eqref{vopt_sigma_neg_CI_0} and \eqref{Q(t_f)>0_cond} yields leads to
    \begin{equation}\label{Z(t_f)}
        Z(t_f) = \frac{2W\sqrt{C_{D,0}C_{D,2}}}{\eta}\left(x_f - x_0\right) - \frac{aQ_0^2}{2} - bQ_0 < 0
    \end{equation}
    Solving \eqref{quadratic_Q(t)} for $Q(t)$ yields the roots
    \begin{equation}\label{Q(t)_pm}
        Q(t) = \frac{-b \pm \sqrt{b^2 - 2aZ(t)}}{a}
    \end{equation}
    If \eqref{Q(t_f)>0_cond} holds, then the discriminant of \eqref{Q(t)_pm} is positive and exceeds $b^2$. One solution of \eqref{Q(t)_pm} is positive and the other is negative.
    Since the battery charge must remain positive and provided that \eqref{Z(t_f)} is satisfied, the admissible solution corresponds to the positive root
    \begin{equation*}
        Q(t) = \frac{-b + \sqrt{b^2 - 2aZ(t)}}{a}
    \end{equation*}
    Since $Q(t)$ is a strictly decreasing function over the interval $[t_0,t_f]$, to satisfy the constraint $Q(t) > Q_{min},\ \forall t\in[t_0,t_f]$, the following must hold
    \begin{equation*}
        Q(t_f) > Q_{min}
    \end{equation*}
    Evaluating $Q(t)$ at $t_f$ yields
    \begin{equation*}
        Q(t_f) = \frac{-b + \sqrt{b^2 - 2aZ(t_f)}}{a}
    \end{equation*}
\end{proof}
\begin{remark}
    The optimal cruise airspeed given by \eqref{vopt_CI=0} in Theorem 1 coincides with the optimal cruise airspeed obtained in \cite{kaptsov2017flight} when the cost of time is equal to zero. Furthermore, when $a = 0$, the supply voltage becomes independent of the battery charge, and the model for $\dot{Q}$ reduces to that used in \cite{kaptsov2017flight}. In this case, the expression for $Q(t_f)$ derived in Theorem 1 corresponds to the solution that would be obtained from the model in \cite{kaptsov2017flight} when the time cost is set to zero. This agreement with the results of \cite{kaptsov2017flight} provides a validation of Theorem 1.
\end{remark}
\section{Simulation Results}\label{simulations}
This section is based on the BETA Technologies ALIA CX300 all-electric aircraft \cite{beta_aircraft}. The battery system is assumed to use Amprius Silicon batteries \cite{stefan2023amprius}, which have a reported energy density of 500 Watt-hours per kilogram. The aircraft and battery parameters are listed in Table \ref{tab:aircraft_params}. The environmental parameters and boundary conditions are summarized in Table \ref{tab:env_params}, which corresponds to a regional advanced air mobility scenario for a regional flight between Montreal, Canada and Ottawa, Canada. 

\begin{table}
    \centering
    \caption{Aircraft Parameters \cite{beta_aircraft}\cite{stefan2023amprius}\cite{streeter2025electric_aviation_beta}\cite{raymer2018aircraft}}
    \begin{tabular}{ll}\hline
        \textbf{Parameter} &\textbf{Value}  \\\hline
        Wing surface area $S$ &30 m$^2$\\
        Profile drag coefficient $C_{D,0}$&0.02\\
        Induced drag coefficient $C_{D,2}$&0.05\\
        Electrical system efficiency $\eta$&0.85\\
        Voltage model coefficient $a$&0.00028 V/C\\
        Voltage model coefficient $b$&682 V\\
        Battery system weight&5,000 N\\
        Maximum charge of batteries $Q_{full}$&979,200 C\\
        Maximum lift coefficient $C_{L,max}$&1.8\\
        Maximum take-off weight (MTOW)&28,675 N\\
        Weight without payload&22,400 N\\
        Maximum rated speed $v_{max,rtd}$&78.6 m/s\\
        Drag divergence speed $v_{dir}$&205.8 m/s\\\hline 
    \end{tabular}
    
    \label{tab:aircraft_params}
\end{table}

\begin{table}
    \centering
    \caption{Environmental Parameters and Boundary Conditions}
    \begin{tabular}{ll}\hline
        \textbf{Parameter} &\textbf{Value}  \\\hline
        Cruising altitude&1,500 m\\
        Air density $\rho$&1.058 kg/m$^3$\\
        Stall speed (at MTOW) $v_{stall}$&31.7 m/s\\
        Initial cruise time $t_0$&0 s\\
        Initial horizontal position $x_0$& 0 m\\
        Initial charge $Q_0$&700,000 C\\
        Aircraft weight $W$& 28,000 N\\
        Final horizontal position $x_f$&150,000 m\\
        Charge lower bound $Q_{min}$&196,000 C\\
        Charge upper bound $Q_{max}$&781,000 C\\
        Airspeed upper bound $v_{max}$&78.6 m/s\\\hline
    \end{tabular}
    
    \label{tab:env_params}
\end{table}
\subsection{Influence of Cruising Altitude and Aircraft Weight on Optimal Cruise Airspeed}
The aircraft weight $W$ is varied over 7 equally spaced values in the interval $[22500,28500]$ Newtons to study its influence on the optimal cruise airspeed $v^*$ as a function of cruising altitude $h$. The analysis considers cruising altitudes ranging from 1000 meters to 4000 meters, representing typical operating altitudes for regional advanced air mobility missions. The air density $\rho$ for each cruising altitude is calculated using the NASA tropospheric model \cite{nasaatmosphere},
\begin{equation}\label{airdensity}
    \rho(h) = \frac{101.29\left(288.14- 0.00649h\right)^{4.256}}{0.2869(288.08)^{5.256}}
\end{equation}
 For each pair $(W,\rho)$, the optimal cruise airspeed is computed using equation \eqref{vopt_CI=0}. Other parameters of the aircraft and of the environment are listed in Tables \ref{tab:aircraft_params} and \ref{tab:env_params}.
\begin{figure}[t]
    \centering
    \includegraphics[scale=0.526]{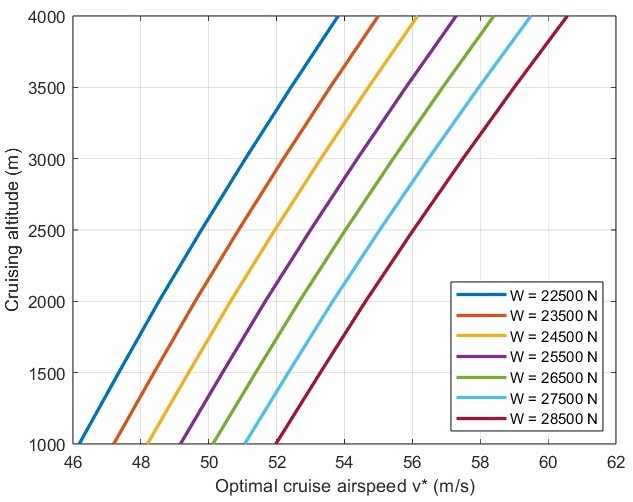}
    \caption{Optimal cruise airspeed $v^*$ as a function of cruising altitude for different values of the aircraft weight $W$.}
    \label{fig:altitude_weight_vs_airspeed}
\end{figure}

Figure \ref{fig:altitude_weight_vs_airspeed} shows that the optimal cruise airspeed $v^*$ increases monotonically with the cruising altitude for all considered values of $W$. This behavior follows directly from the expression for $v^*$. As the aircraft flies higher, the atmospheric air density $\rho$ decreases, and a higher airspeed is therefore required to generate the lift needed to balance weight during steady flight. The figure also illustrates the influence of aircraft weight on $v^*$. For a given cruising altitude, larger values of $W$ result in larger values of $v^*$. This is consistent with the fact that heavier aircraft must operate at higher airspeeds in order to generate the lift required to sustain steady flight. 

These results indicate that energy-optimal cruise operation requires adjusting the cruise airspeed according to both the aircraft weight and the cruising altitude. Since low-altitude AAM missions may be conducted across multiple altitudes and with varying payloads, accounting for these factors when selecting cruise speeds can improve energy efficiency and help ensure that the AAM missions with electric aircraft remain feasible.

\subsection{Influence of Aircraft Weight and Initial Charge on the Minimum Feasible Electrical System Efficiency}
The analysis considers aircraft weights between 22,500 Newtons and 28,500 Newtons to study the influence of $W$ on the minimum electrical system efficiency $\eta$ required to satisfy the feasibility condition $Q(t_f) > Q_{min}$. Several representative values of the initial onboard charge $Q_0$ are considered in order to evaluate the impact of the available battery energy on this feasibility condition. For each pair $(W,Q_0)$, the minimum $\eta$ is computed from the condition $Q(t_f) = Q_{min}$, which represents the boundary between feasible and infeasible all-electric cruise operations. 
\begin{figure}
    \centering
    \includegraphics[scale=0.53]{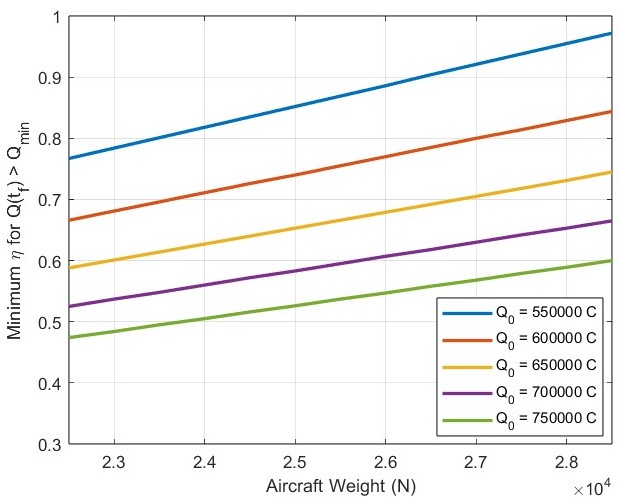}
    \caption{Minimum electrical system efficiency $\eta$ required to satisfy $Q(t_f) > Q_{min}$ as a function of the aircraft weight $W$ for various values of $Q_0$.}
    \label{fig:mineta_vs_weight}
\end{figure}

Figure \ref{fig:mineta_vs_weight} shows that the minimum electrical system efficiency $\eta$ required to satisfy $Q(t_f) > Q_{min}$ increases monotonically with the weight of the aircraft for all values of $Q_0$ considered. This implies that heavier aircraft require higher electrical system efficiencies to ensure that sufficient charge remains available at the end of cruise. This is consistent with the fact that heavier aircraft require greater energy expenditure to sustain flight. The figure also shows that, for a given aircraft weight, the less efficient the electrical system needs to be for larger values of $Q_0$ so that the aircraft can complete its cruise segment. 

From an operational perspective, these results highlight the importance of aircraft weight and available charge in determining the feasibility of all-electric cruise operations. This consideration is especially relevant for AAM missions, where aircraft weight vary significantly due to changes in passenger load or cargo, and where frequent operations may result in varying initial battery charge levels. In such scenarios, heavier payloads require more efficient electrical propulsion systems to maintain feasible all-electric cruise operations, while a larger initial onboard charge relaxes these efficiency requirements.

\section{Conclusions}
This paper investigated the total energy problem for an all-electric aircraft in steady cruise flight. An optimal control formulation is developed in which the battery supply voltage is modeled as an affine function of the battery charge. Closed-form expressions were derived for the optimal cruise airspeed and optimal final time. The analysis showed that the optimal airspeed is the minimum-drag airspeed and remains constant over the cruise segment. Additional analytical conditions were obtained to determine the feasibility of all-electric operations based on the available battery charge.

Numerical simulations based on the BETA Technologies CX-300 all-electric aircraft and a representative advanced air mobility mission between Montreal and Ottawa illustrated the practical implications of the analytical results. The simulations show that the optimal cruise airspeed increases with both cruising altitude and aircraft weight, while the minimum electrical system efficiency required for the aircraft to maintain battery charge above a given threshold increases with aircraft weight and decreases with initial charge. The analysis suggests that energy-optimal cruise operation requires the adjustment of the optimal cruise airspeed according to the aircraft weight and the cruising altitude. It also indicates that the feasibility of sustained all-electric cruise depends on both aircraft weight and available onboard charge. These findings provide valuable insights for AAM mission planning and operational decision-making, where variations in payload, cruising altitude, and available onboard battery can significantly impact the feasibility and energy efficiency of all-electric aircraft operations.

\section*{Acknowledgment}
The authors would like to acknowledge Concordia University for providing the necessary infrastructure which facilitated the successful completion of this work.

\section*{Conflict of Interest}
On behalf of all authors, the corresponding author states that there is no conflict of interest.

\section*{Data Availability}
Data used for Figure \ref{fig:battery_model} of the manuscript comes from the Oxford Battery Degradation Dataset found in \cite{howey2017oxford}. All other data underlying the result are available as part of the article.

\section*{Funding}
The authors received no specific funding for this work.

\section*{Author Contributions}
Steven Li wrote the manuscript. Luis Rodrigues supervised the research, contributed to the conception and design of the method, and reviewed and edited the manuscript.
\bibliography{References}

\end{document}